\documentclass[12pt,twoside,a4paper,leqno]{article}
\usepackage[dvips,%
pdftitle={On the geometry of the energy operator in quantum mechanics},%
pdfauthor={C. Tejero Prieto, R. Vitolo},%
pdfkeywords={Classical mechanics, quantum mechanics, geometric quantization},%
]{hyperref}
\usepackage{amsmath}
\usepackage{amssymb}
\usepackage{amsthm}
\usepackage[mathscr]{eucal}
\usepackage{amscd}
\usepackage[scriptlabels,midshaft]{diagrams}
\usepackage{xcolor}

\newcommand{\myskip}{\vspace*{8pt}}
\DeclareMathAlphabet{\eurm}{U}{eur}{m}{n}
\DeclareMathAlphabet{\eubf}{U}{eur}{b}{n}
\DeclareFontFamily{U}{UWCyr}{}
\DeclareFontShape{U}{UWCyr}{m}{n}{%
  <5> <6> <7> <8> <9>
  <10> <10.95> <12> <14.4> <17.28> <20.74> <24.88> wncyr10
  }{}
\DeclareFontShape{U}{UWCyr}{m}{it}{%
  <5> <6> <7> <8> <9>
  <10> <10.95> <12> <14.4> <17.28> <20.74> <24.88> wncyi10
  }{}
\DeclareFontShape{U}{UWCyr}{m}{sc}{%
  <5> <6> <7> <8> <9>
  <10> <10.95> <12> <14.4> <17.28> <20.74> <24.88> wncysc10
  }{}
\DeclareFontShape{U}{UWCyr}{b}{n}{%
  <5> <6> <7> <8> <9>
  <10> <10.95> <12> <14.4> <17.28> <20.74> <24.88> wncyb10
  }{}
\DeclareMathAlphabet{\cyrm}{U}{UWCyr}{m}{n}
\DeclareMathAlphabet{\cyit}{U}{UWCyr}{m}{it}
\DeclareMathAlphabet{\cysc}{U}{UWCyr}{m}{sc}
\DeclareMathAlphabet{\cybf}{U}{UWCyr}{b}{n}

\newcounter{assump}
\newtheorem{Assumption}{\indent Assumption}[assump]
\newcounter{postul}
\newtheorem{Postulate}{\indent Postulate}[postul]

\newtheoremstyle{MyThm}
  {3pt}
  {3pt}
  {\itshape}
  {\parindent}
  {\bfseries}
  {.}
  {.5em}
  {}
\swapnumbers
\theoremstyle{MyThm}
\newtheorem{Definition}{Definition}[section]

\newtheorem{Caution}[Definition]{Caution}
\newtheorem{Convention}[Definition]{Convention}
\newtheorem{Corollary}[Definition]{Corollary}
\newtheorem{Example}[Definition]{Example}
\newtheorem{Exercise}[Definition]{Exercise}
\newtheorem{Lemma}[Definition]{Lemma}
\newtheorem{Notation}[Definition]{Notation}
\newtheorem{Note}[Definition]{Note}
\newtheorem{Problem}[Definition]{Problem}
\newtheorem{Proposition}[Definition]{Proposition}
\newtheorem{Remark}[Definition]{Remark}
\newtheorem{Theorem}[Definition]{Theorem}
\newcommand{\bAs}{\begin{Assumption}\em}
\newcommand{\eAs}{\end{Assumption}}
\newcommand{\bCa}{\begin{Caution}\em}
\newcommand{\eCa}{\end{Caution}}
\newcommand{\bCr}{\begin{Corollary}\em}
\newcommand{\eCr}{\end{Corollary}}
\newcommand{\bCv}{\begin{Convention}\em}
\newcommand{\eCv}{\end{Convention}}
\newcommand{\bDf}{\begin{Definition}\em}
\newcommand{\eDf}{\end{Definition}}
\newcommand{\bDr}{\begin{Exercise}\em}
\newcommand{\eDr}{\end{Exercise}}
\newcommand{\bEx}{\begin{Example}\em}
\newcommand{\eEx}{\end{Example}}
\newcommand{\bLm}{\begin{Lemma}\em}
\newcommand{\eLm}{\end{Lemma}}
\newcommand{\bNo}{\begin{Notation}\em}
\newcommand{\eNo}{\end{Notation}}
\newcommand{\bNt}{\begin{Note}\em}
\newcommand{\eNt}{\end{Note}}

\newcommand{\bPb}{\begin{Problem}\em}
\newcommand{\ePb}{\end{Problem}}
\newcommand{\bPf}{\begin{proof}[\noindent\indent{\sc Proof}]}
\newcommand{\ePf}{\renewcommand{\qedsymbol}{}\end{proof}}

\newcommand{\bPr}{\begin{Proposition}\em}
\newcommand{\ePr}{\end{Proposition}}
\newcommand{\bPs}{\begin{Postulate}\em}
\newcommand{\ePs}{\end{Postulate}}
\newcommand{\bRm}{\begin{Remark}\em}
\newcommand{\eRm}{\end{Remark}}

\newcommand{\bTh}{\begin{Theorem}}
\newcommand{\eTh}{\end{Theorem}}
\newcommand{\bEq}{\begin{eqnarray}}
\newcommand{\eEq}{\end{eqnarray}}
\newcommand{\beq}{\begin{eqnarray*}}
\newcommand{\eeq}{\end{eqnarray*}}
\newcommand{\bCd}{\bEq\begin{CD}}
\newcommand{\eCd}{\end{CD}\eEq}
\newcommand{\bcd}{\beq\begin{CD}}
\newcommand{\ecd}{\end{CD}\eeq}
\newcommand{\bdg}{\beq\begin{diagram}}
\newcommand{\edg}{\end{diagram}\eeq}
\newcommand{\bDg}{\bEq\begin{diagram}}
\newcommand{\eDg}{\end{diagram}\eEq}
\newcommand{\ben}{\begin{enumerate}}
\newcommand{\een}{\end{enumerate}}
\newcommand{\btb}{\begin{tabbing}}
\newcommand{\etb}{\end{tabbing}}
\newcommand{\bsm}{\begin{quotation}\small}
\newcommand{\esm}{\end{quotation}}
\newcommand{\bfz}{\begin{footnotesize}}
\newcommand{\efz}{\end{footnotesize}}
\newcommand{\bsz}{\begin{scriptsize}}
\newcommand{\esz}{\end{scriptsize}}

\newcommand{\fz}{\footnotesize}

\newcommand{\Rn}{{I\!\!R}}
\newcommand{\Cn}{{\B C}}

\newcommand{\h}{\hbar}

\newcommand{\1}{\mathbf 1}

\newcommand{\imi}{{\eurm{i}\,}}

\newcommand{\coi}{{\mathfrak{i}\,}}

\newcommand{\der}{\partial}
\newcommand{\nab}{\nabla}

\newcommand{\Upa}{^{\uparrow}{}}

\newcommand{\Nat}{^{\natural}{}}

\newcommand{\Fla}{^{\flat}{}}
\newcommand{\Sha}{^{\sharp}{}}

\newcommand{\Prl}{^{\|}{}}

\newcommand{\mto}{\mapsto}

\newcommand{\com}{\circ}
\newcommand{\car}{\times}
\newcommand{\ten}{\otimes}

\newcommand{\wed}{\wedge}

\DeclareMathOperator{\con}{\lrcorner}

\newcommand{\eqv}{\,\equiv\,}

\DeclareMathOperator{\byd}{\,{\raisebox{.1ex}{$\eurm :$}{\eurm =}}\,}

\newcommand{\ucar}[1]{\underset{#1}{\times}}

\newcommand{\uten}[1]{\underset{#1}{\otimes}}

\newcommand{\db}[1]{{\,{#1}\!{#1}\,}}

\newcommand{\tfr}[2]{\tfrac{#1}{#2}\,}

\newcommand{\rtd}[1]{\sqrt{|#1|}}

\newcommand{\di}[2]{{\frac{\der_{#1}\sqrt{|#2|}}{\sqrt{|#2|}}}}

\newcommand{\diG}[3]
{{\frac{\der_{#1}(G^{#1#2}_0\sqrt{|#3|})}{\sqrt{|#3|}}}}

\newcommand{\END}{{\,\text{\footnotesize\qedsymbol}}}

\DeclareMathOperator{\id}{{id}}

\newcommand{\bAl}
{\vspace{-0.8cm}
\begin{alignat*}{2}
& \qquad\qquad\qquad\qquad\qquad\qquad\qquad\qquad\qquad\qquad
&&\qquad\qquad\qquad\qquad\qquad\qquad\qquad\qquad\qquad
\\}

\newcommand{\f}[1]{{\boldsymbol{#1}}}

\newcommand{\ba}[1]{{{\bar{#1}}}}

\newcommand{\ch}[1]{{\check{#1}}}
\newcommand{\wch}[1]{{\overset{\vee}{#1}}}

\newcommand{\dt}[1]{{\dot{#1}}}

\newcommand{\ob}[1]{{\overset{o}{#1}}}

\newcommand{\E}[1]{{\eurm{#1}}}

\newcommand{\C}[1]{{\mathcal{#1}}}

\newcommand{\B}[1]{{\mathbb{#1}}}

\newcommand{\K}[1]{{\cyrm{#1}}}

\newcommand{\alp}{\alpha}

\newcommand{\gam}{\gamma}

\newcommand{\lam}{\lambda}

\newcommand{\Del}{\Delta}

\newcommand{\The}{\Theta}

\newcommand{\Lam}{\Lambda}

\newcommand{\Ome}{\Omega}

\newcommand{\ke}{\chi}
\newcommand{\Kin}[2]{\tfrac12 \, G^0_{#1#2} \, x^{#1}_0 \, x^{#2}_0}
\newcommand{\Mom}[2]{G^0_{#1#2} \, x^{#2}_0}
\title{\bf On the geometry of the energy operator
\\
in quantum mechanics}
\author{
\bf Carlos Tejero Prieto$^1$, Raffaele Vitolo$^2$
\bigskip
\\
\fz $^1$ Departamento de Matematicas, Universidad de Salamanca
\\
\fz Pl. de la Merced 1--4, 37008 Salamanca, Spain
\\
\fz email: {\tt carlost@usal.es}
\myskip
\\
\fz $^2$Dipartimento di Matematica e Fisica ``E. De Giorgi, Universit\`a del
Salento
\\
\fz Via per Arnesano, 73100 Lecce, Italy
\\
\fz email: {\tt raffaele.vitolo@unisalento.it}
}
\date{{\small\itshape We dedicate this paper to the 70th birthday\\ of Luigi Mangiarotti and
  Marco Modugno}}
\begin{document}
\maketitle
\begin{abstract}
  We analyze the different ways to define the energy operator in geometric
  theories of quantum mechanics. In some formulations the operator contains the scalar
  curvature as a multiplicative term. We show that such term can be canceled
  or added with an arbitrary constant factor, both in the mainstream Geometric
  Quantization and in the Covariant Quantum Mechanics, developed by Jadczyk and
  Modugno with several contributions from many authors.
\end{abstract}

\section{Introduction}

One of the problems of quantum mechanical theories is the fact that it is not
possible to consistently quantize all physical observables. This fact finds its
justification in different ways: the Heisenberg principle forbids the
simultaneous localization of position and momenta observables, it is not
possible to find an irreducible representation of the space of all polynomials
in position and momenta (Groenewold--Van Hove's theorem), etc..

The problem persists in the mathematical models of quantum mechanical theories. 
Within such models, one of the most developed and successful is Geometric Quantization 
(GQ for short, see for example \cite{Sni80,Woo92}). In this theory it is possible to 
quantize the family of observables which preserve the directions of another distinguished
family of observables. The space of such directions is an integrable
lagrangian distribution on the phase space, and it is said to be a polarization.  For
example, in the Schr\"odinger quantization of particle mechanics it is possible
to quantize observables which are linear in the momenta, that is of the form $f^i(x^j)p_i + f_0(x^j)$, since
their Hamiltonian vector fields preserve the tangent vectors $\der/\der p_i$,
which span the so-called vertical polarization. However, the energy is a quadratic function of the
momenta and breaks this prescription.

A similar phenomenon occurs in a more recent geometric framework of
quantum mechanics that implements the principle of relativity, in the sense of invariance with
respect to changes of observer or reference frame. This theory
is called Covariant Quantum Mechanics (CQM).
It was initiated by Jadczyk and Modugno
\cite{JadMod92,JadMod93,JadMod94} and later developed by several other authors
\cite{Can10,JadJanMod98,JanMod02,ModSalTol06,ModTejVit00,ModTejVit08a,
  ModTejVit08b,ModVit96,SalVit00,Tej01,Vit98,Vit00}, also with extensions to
general relativistic mechanics \cite{JanMod08,JanMod09,JanVit12}. 

In this paper we compare the peculiar ways of quantizing energy in GQ and
CQM. In this process we find out interesting geometric features of the
covariant energy operator in the two theories. Let us discuss this program more
in detail.

In GQ the standard way to quantize the energy in a certain polarization is the so-called
Blattner-Kostant-Sternberg method (BKS). This method is applicable to those
observables, like the energy, whose Hamiltonian vector fields do not preserve the polarization. One starts with a wave function polarized along the chosen polarization, then both data are dragged by the flow of the energy vector field in order to produce new wave functions polarized along infinitesimally close polarizations. Under suitable conditions the wave functions polarized with respect to the original and the new polarizations are related by the so called BKS-pairing. By means of it one is able to express the change induced by dragging the wave function as a one parameter family of wave functions polarized along the original polarization. The zero time derivative of this family of wave functions gives, by definition, the action of the quantized energy operator on the initial wave function. This yields the method for obtaining a `correct' quantization of energy in GQ.

In CQM we define two quantum operators connected with energy: 1 - the Lie
derivative of wave functions with respect to the Hamiltonian vector field of the
energy; 2 - a Schr\"odinger operator obtained from a Lagrangian which is
uniquely characterized by covariance requirements. Then, the quantization of
energy is the only linear combination of the above operators that does not
depend on time derivatives. Section~\ref{Covariant Quantum Mechanics} contains
a summary of CQM with an emphasis on concepts and formulae which are
relevant to the definition of the energy operator.

In GQ the wave functions are defined as \emph{half-forms}, \emph{i.e.} sections
of a complex line bundle twisted by the square root of the
bundle of volume forms normal to the polarization. This enables one to integrate the
natural pairing between any two half-forms. Half-forms were introduced in
\cite{Bla73}. Initially, \cite{JadMod92,JadMod93,JadMod94}  CQM also used
half-forms for defining wave functions, then they have been dropped by assuming a
Hermitian metric with values in densities \cite{JadJanMod98,JanMod02,JanMod06}. In order to ease the comparison with GQ, here we use CQM as originally formulated with half-forms.

If we wish to compare GQ and CQM we should make further assumptions on GQ's
general setting. In particular, such a comparison makes sense if we consider
systems of particles, possibly with holonomic constraints, in the Schr\"odinger
representation. In other words, we consider a Riemannian manifold $\f M$ and
the symplectic manifold $T^*M$, and quantize with respect to the vertical
polarization, which is generated by the vector fields $\der/\der p_i$.

The quantization of energy in GQ and CQM leads to energy operators which differ
by a term which is a multiplication operator on the wave function by the scalar
curvature of the given spacelike metric. Let us discuss this feature more in
detail.

The scalar curvature term first made its appearance in the paper \cite{dew}, in
the context of Feynman path integral approach to quantum mechanics. Several
authors have tried to determine the factor in front of the scalar curvature
term via path integral, but they have found different results, according to
different ways of performing the integral. Similar computations within the BKS
approach, together with their relation to path integral, are presented in
\cite{Sni80,Woo92} and exhibit a scalar curvature multiplication operator of
the form
\begin{equation}
\Psi \mapsto \textstyle\frac{1}{6}r\Psi\label{eq:1},
\end{equation}
where $r$ is the scalar curvature of the Riemannian metric of the configuration
space. It shall be remarked that not all authors in GQ show this term in their
computations.

According with CQM, scalar curvature can be added to the Schr\"odinger operator
through its Lagrangian as a constant multiple of the norm of sections of the
quantum bundle. The corresponding Euler--Lagrange expressions contain the term
\begin{equation}
\Psi \mapsto kr\Psi\label{eq:2},
\end{equation}
where $k\in\Rn$.  It was proved by covariance arguments
that such a Lagrangian is unique \cite{Jan01} and that the corresponding
Schr\"odinger operator is also unique \cite{JanMod02} exactly up to the constant
factor $k$ in front of the scalar curvature, that remains undetermined.

The main result of this paper (Section~\ref{sec:energy-operator-from}) is that
the scalar curvature operator in GQ is canceled out when one performs the double
covariant derivative in the Bochner Laplacian by keeping into account that
there is a natural connection on the square root of the vertical
polarization. Indeed, polarized sections are of the form $\Psi\colon \f M \to
\f Q\otimes \sqrt{\wedge^nT^*\f M}$ (with a possible dependence on time), where
$\f Q \to \f M$ is a Hermitian complex line bundle and $n=\dim \f M$. Since $\f
M$ is a Riemannian manifold there is a natural connection on the bundle
$\wedge^nT^*\f M \to \f M$. It is natural to assume that in Bochner's Laplacian
$g^{ij}\nabla_i\nabla_j$ covariant derivatives are tensor products of the
connection on $\f Q$ which is required by the quantization process and the
connection on the square root bundle of densities.

On the other hand, a multiplication operator of the type~\eqref{eq:2} may
always be added to the energy operator without violating covariance.
To our knowledge there is no evidence in experiments of the presence of the
scalar curvature term in the energy operator. In view of our results and of the
previous results about the indeterminacy of the constant $k$ we argue that
it might be the case that such a term has simply no physical
relevance.

\section{Covariant Quantum Mechanics}
\label{Covariant Quantum Mechanics}

We start with a summary of the classical and quantum theory developed by
Jadczyk, Jany\v ska and Modugno.  The interested reader may refer, for
instance, to \cite{JadMod94,JadJanMod98,JanMod02,JanMod06} for further details.

In CQM `covariance' includes also independence from the choice of
units of measurements. For this reason, we developed a rigorous treatment of
spaces uf units of measurement; roughly speaking they have the same algebraic
structure of $\Rn^+$, but no distinguished generator over $\Rn^+$
\cite{JanModVit10}. In this paper, we assume the following ``positive
1--dimensional semi--vector spaces" over $\Rn^+$ as fundamental unit spaces:
the space $\B T$ of {\em time intervals\/}, the space $\B L$ of {\em
  lengths\/}, the space $\B M$ of {\em masses\/}. Moreover, we assume the {\em
  Planck constant\/} to be an element $\h \in \B T ^* \ten \B L^2 \ten \B M$.
We refer to a particle with mass $m \in \B M$ and charge $q \in \B T^* \ten \B
L^{3/2} \ten \B M^{1/2}$.

\subsection{Classical theory}
\label{Classical theory}

The {\em spacetime\/} is an oriented $(n+1)$--dimensional manifold $\f E$ (in
the standard case $n = 3$), the {\em absolute time\/} is an affine space
associated with the vector space $\Rn \ten \B T$, the {\em absolute time map\/}
is a fibring $t : \f E \to \f T$. We denote fibred charts of spacetime by
$(x^\lam) \equiv (x^0, x^i)$; the corresponding vector fields and forms are
denoted by $\der_0$, $\der_i$ and $d^0$, $d^i$. The tangent space and the
vertical space of $\f E$ are denoted by $T\f E$ and $V\f E$.  It is easy to
check that $\f E$ is orientable if and only if it is spacelike-orientable; that
means, that $\wedge^4 T^*\f E\to \f E$ is a trivial line bundle if and only if
$\wedge^3 V^*\f E\to \f E$ is a trivial line bundle. As usual, $V\f E\byd \ker
T\f E$.

A {\em motion\/} is a section $s : \f T \to \f E$. The {\em phase space\/} is
the first jet space of motions $J_1\f E$ (see \cite{KolMicSlo93,ManMod83} about
jet spaces).  We denote fibred charts of phase space by $(x^0, x^i;
x^i_0)$. The {\em absolute velocity\/} of a motion $s$ is its first jet
prolongation $j_1s : \f T \to J_1\f E$. An {\em observer\/} is a section $o :
\f E \to J_1\f E$ and the {\em observed velocity\/} of a motion $s$ is the map
$\nab[o] s \byd j_1s - o \com s : \f T \to \B T^* \ten V\f E$.

The {\em spacelike metric\/} is a scaled Riemannian metric of the fibres of
spacetime $g : \f E \to \B L^2 \ten (V^*\f E \uten{\f E} V^*\f E)$. Given a
particle of mass $m$, it is convenient to consider the re--scaled spacelike
metric $G \byd \tfrac{m}{\h}  g : \f E \to \B T \ten (V^*\f E \uten{\f E}
V^*\f E)$. The spacelike and spacetime volumes are, respectively, the tensor
fields
\begin{equation}
  \eta \colon \f E \to  \B L^3 \ten \wedge^3 V^*\f E,\quad
  \eta = \sqrt{|g|}\check{d}^1\wedge\check{d}^2\wedge\check{d}^3,\quad
  \bar{\eta}\colon \f E \to \B T \ten \B L^3 \ten \wedge^4 T^*\f E,\quad
  \bar{\eta} = dt\wedge\eta,
\end{equation}
where $|g|=\det(g_{ij})$ and $\check{}$ denotes vertical restriction.

The {\em gravitational field\/} is a time preserving torsion free linear
connection on the tangent bundle of spacetime $K\Nat : T\f E \to T^*\f E
\uten{T\f E} TT\f E$, such that $\nab[K\Nat] g = 0$ and the curvature tensor
$R[K\Nat]$ fulfills the condition $R\Nat{}_\lam{}^i{}_\mu{}^j =
R\Nat{}_\mu{}^j{}_\lam{}^i$. The `metricity' condition on $K\Nat$ implies that
its vertical restriction coincides with the family of Riemannian connections
induced by $g$ on the fibres of $\f E\to \f T$.  This implies that the
Christoffel symbols with three spacelike indexes are of the type $\Gamma^i_{jk}
= -\tfr12 g^{ih}(\der_jg_{hk} + \der_kg_{hj} - \der_hg_{jk})$\footnote{Note the
  difference in sign with respect to the standard convention.}.

The {\em electromagnetic field\/} is a scaled 2--form $f : \f E \to (\B L^{1/2}
\ten \B M^{1/2}) \ten \Lam^2 T^*\f E$, such that $df = 0$. Given a particle of
charge $q$, it is convenient to consider the re--scaled electromagnetic field
$F \byd \tfrac{q}{\h}  f : \f E \to \Lam^2 T^*\f E$.

The electromagnetic field $F$ can be ``added'', in a covariant way, to
the gravitational connection $K\Nat$ yielding a {\em (total)
spacetime connection\/} $K$, with coordinate expression
\beq
K_i{}^h{}_j = K\Nat{}_i{}^h{}_j ,
\quad
K_j{}^h{}_0 = K_0{}^h{}_j = K\Nat{}_0{}^h{}_j + \frac{1}{2}  F^h{}_j
,
\quad
K_0{}^h{}_0 = K\Nat{}_0{}^h{}_0 + \frac{1}{2}  F^h{}_0 .
\eeq
This turns out to be a time preserving torsion free linear connection
on the tangent bundle of spacetime, which still fulfills the properties
that we have assumed for $K\Nat$.

The spacetime fibration, the total spacetime connection and
the spacelike metric, yield, in a covariant way, a 2--form
$\Ome : J_1\f E \to \Lam^2 T^*J_1\f E$ on the phase space, with coordinate
expression
\begin{equation}\label{eq:3}
\Ome = G^0_{ij} 
\big(d^i_0 - (K_\lam{}^i{}_0 + K_\lam{}^i{}_h  x^h_0)
d^\lam\big) \wed (d^j - x^j_0  d^0) .
\end{equation}
This is a {\em cosymplectic form\/}
(see \cite{JanMod09} and refences therein for a deeper discussion on the
geometry of these objects), i.e. it fulfills
the following properties: 1) $d \Ome = 0$, 2)
$dt \wed \Ome^n : J_1\f E \to \B T \ten \Lam^nT^*J_1\f E$ is a scaled
volume form on $J_1\f E$. Conversely, the cosymplectic form
$\Ome$ characterises the spacelike metric and the total spacetime
connection. Moreover, the closedness of $\Ome$ is equivalent to the
conditions that we have assumed on $K$.

There is a unique second order connection \cite{ManMod83} $\gam : J_1\f E \to
\B T^* \ten TJ_1\f E$, such that $i_\gam \Ome = 0$.  We assume the generalised
{\em Newton's equation\/} $\nab[\gam] j_1s = 0$ as the equation of motion for
classical dynamics. Of course the above equation also admits a Lagrangian and a
Hamiltonian formulation \cite{ModVit96}. The cosymplectic form $\Ome$ admits
locally potentials of the type $\The : J_1\f E \to T^*\f E$. It turns out that
these potentials are the {\em Poincar\'e--Cartan forms\/} of the Lagrangians
$\C L$ that can be obtained as one of the two summands of the splitting $\The =
\C L + \C P$, where $\C L : J_1\f E \to T^*\f T$ and $\C P : J_1\f E \to V^*\f
E$ is the {\em momentum\/}.  These components are observer independent, but
depend on the chosen gauge of the starting Poincar\'e--Cartan form. On the
other hand, given an observer $o$, each Poincar\'e--Cartan form $\The$ splits,
according to the splitting of $T^*\f E$ induced by $o$, into the horizontal
component $- \C H[o] : J_1\f E \to T^*\f T$, which is called the observed {\em
  Hamiltonian\/}, and the vertical component $\C P[o] : J_1\f E \to V^*\f E$,
which is the observed {\em momentum\/}. We have the coordinate expressions
\begin{equation}
\C L =
\left(\frac{1}{2} G^0_{ij} x^i_0 x^j_0 + A_i x^i_0 + A_0\right) d^0,
\quad \C P = (G^0_{ij} x^j_0 + A_i) (d^i - x^i_0 d^0) ,\label{eq:4}
\end{equation}
and, in a chart adapted to $o$,
\begin{equation}
\C H[o] = \C H_0d^0 = \left(\frac{1}{2}  G^0_{ij}  x^i_0
x^j_0 - A_0\right)  d^0,
\quad \C P[o] = \C P_j d^j = (G^0_{ij}  x^j_0 + A_i)  d^i ,\label{eq:5}
\end{equation}
where $A \equiv o^* \The$.

The cosymplectic form $\Ome$ yields in a covariant way the
Hamiltonian lift of functions $f: J_1\f E \to \Rn$ to vertical vector
fields $H[f] : J_1\f E \to VJ_1\f E$; consequently, we obtain the
Poisson bracket $\{f,g\}$ between functions of phase space. Given an
observer, the law of motion can be expressed, in a non covariant way,
in terms of the Poisson bracket and the Hamiltonian.

More generally, chosen a time scale $\tau : J_1\f E \to T\f T$, the
cosymplectic form $\Ome$ yields, in a covariant way, the Hamiltonian lift of
functions $f$ of phase space to vector fields $H_\tau[f] : J_1\f E \to TJ_1\f
E$, whose time component is $\tau$. In particular, let us introduce the
cosymplectic isomorphism $\Ome\Fla\colon T_\tau J_1\f E \to T^*_\gamma J_1\f E$
between the subspace of vectors in $TJ_1\f E$ that project to $\tau$ and the
subspace of one-forms in $T^*J_1\f E$ that annihilate $\gamma$.  Let us
denote by $\Ome\Sha_\tau$ the inverse of $\Ome\Fla_\tau$.  Then we define
$H_\tau[f]=\Ome\Sha_\tau(df-\gamma . f)$. It can be proved that $H_\tau[f]$ is
projectable onto a vector field $X[f] : \f E \to T\f E$ if and only if the
following conditions hold: i) the function $f$ is quadratic with respect to the
affine fibres of $J_1\f E \to \f E$ with second fibre derivative $f'' \ten G$,
where $f'' : \f E \to \Rn \otimes \B T$, ii) $\tau = f''$. A function of this
type is called a \emph{special phase function} and has coordinate expression of
the type
\begin{equation}
f = \frac{1}{2} f^0  G^0_{ij} x^i_0 x^j_0 + f^0_i  x^i_0 + f_0 , \qquad
{\rm with} \qquad f^0, f^0_i, f_0 : \f E \to \Rn .\label{eq:6}
\end{equation}
Note that $f''=f^0u_0$, where $u_0\in \B T$ is a time scale. From now on we
will assume that $f^0$ is a constant, even if this assumption could be dropped
\cite{JadMod94}.

The vector space of special phase functions is not closed under
the Poisson bracket, but it turns out to be an $\Rn$--Lie algebra
through the covariant {\em special bracket\/}
\begin{equation}
\db[ f, g\db] = \{f,g\} + \gam(f'')\cdot g - \gam(g'')\cdot f .\label{eq:7}
\end{equation}
Moreover, the map $f \mto X_f$ turns out to be a morphism of Lie algebras; we
have the coordinate expression $X_f = f^0  \der_0 - f^i  \der_i$, where
$f^i=G^{ij}f_j$.
\subsection{Quantum theory}
\label{Quantum theory}

Let us consider a  complex
line bundle over spacetime $\f Q \to \f E$ equipped with a
Hermitian metric $h : \f Q \ucar{\f E} \f Q \to \Cn$.
We shall refer to normalised local bases $b$ of $\f Q$ and to the
associated complex coordinates $z$; accordingly, the coordinate
expression of a \emph{local section} is of the type $\Psi = \psi
 b$, with $\psi : \f E \to \Cn$.

We consider also the \emph{extended line bundle} $\f Q\Upa \to J_1\f E$,
$\f Q\Upa \byd \f Q \ucar{\f E} J_1\f E$. A family (or `system') of connections
of $\f Q$ parametrised by observers $o\colon \f E \to J_1\f E$ induces, in a
covariant way, a connection of $\f Q\Upa$, which is called {\em universal\/}
\cite{ManMod83,JadJanMod98}. A characteristic property of the universal
connection is that its contraction with any vertical vector field of the bundle
$J_1\f E \to \f E$ vanishes; in coordinates, $\K q ^0_i = 0$.

It is well known that the Picard group $\text{Pic}(M)$ of isomorphism classes
of complex line bundles over a differentiable manifold $M$ can be identified
with the second integral cohomology group $H^2(M,\mathbb Z)$ by means of the
first Chern class mapping $$c_1\colon \text{Pic}(M)\xrightarrow{\sim}
H^2(M,\mathbb Z),$$ such that given a Line bundle $L\to M$ sends its
isomorphism class $[L]\in \text{Pic}(M)$   to the first Chern class $c_1(L)\in
H^2(M,\mathbb Z)$. Since $J_1\f E\to \f E$ is an affine bundle, it follows from
elementary obstruction theory that the map induced on cohomology by pullback
along the map $J_1 \f E\to \f E$ is an isomorphism $(\cdot)\Upa\colon H^2(\f
E,\mathbb Z)\xrightarrow{\sim} H^2(J_1\f E,\mathbb Z)$. Therefore, by pulling
back complex line bundles on $\f E$ to $J_1\f E$ we also get an isomorphism of
Picard groups $(\cdot)\Upa\colon \text{Pic}(\f E)\xrightarrow{\sim}
\text{Pic}(J_1\f E)$.

We say that $\f Q \to \f E$ is a \emph{quantum bundle} if there exists a
connection $\K q$: $\f Q\Upa \to T^*J\f E \uten{J\f E} T\f Q\Upa$ on the
extended quantum bundle, called a \emph{quantum connection}, which is
Hermitian, universal and whose curvature is {$ R[\K q] = \imi \Ome \ten
  \id_{\f Q}$.} We stress that $\tfrac{1}{\h}$ has been incorporated in $\Ome$
through the re--scaled metric $G$. In a local base $b$, a quantum connection
$\K q$ is of the type \bEq \K q = \K q\Prl + \coi \The \ten \1 , \eEq where $\K
q\Prl$ is the flat connection associated with $b$ and $\The$ is a potential of
$\Ome$, the Poincar\'e--Cartan form. Given an observer $o$ we can also write
$\K q = \K q\Prl + \coi (- \C H[o] + \C P[o]) \ten \1$. We have the coordinate
expression
\begin{multline}\label{eq:8}
\K q = d^\lam \ten \der_\lam + d^i_0 \ten \der^0_i +
     \coi  \K q_\lam  d^\lam \ten (z  \der z)
\\
   = d^\lam \ten \der_\lam + d^i_0 \ten \der^0_i +
     \coi  \big(- (\Kin ij - A_0)  d^0
+
	  (\Mom ij + A_i)  d^i \big) \ten (z  \der z).
\end{multline}
Given an observer $o$ we have the {\em observed quantum connection\/}
$o^*\K q : \f Q \to T^*\f E \uten{\f E} \f Q$
with coordinate expression $o^*\K q = d^\lam \ten \der_\lam + \coi A_\lam  z
 d^\lam \ten \der z$ where $A_\lam d^\lam = o^*\The$.

A quantum connection exists if and only if the cohomology class of $\Ome$ is
integral; the equivalence classes of quantum bundles equipped with a quantum
connection are classified by the cohomology group $H^1\big(\f E, U(1)\big)$
\cite{ModVit96,Vit98}.

In what follows we assume a quantum bundle equipped with a quantum connection.

Any other quantum object is obtained, in a covariant way, from this
quantum structure. The quantum connection is defined on the extended
quantum bundle, while we are looking for further quantum objects
living on the original quantum bundle. This goal is successfully
achieved by a {\em method of projectability\/}: namely, we look for
objects of the extended quantum bundle which are projectable to the
quantum bundle and then we take their projections. Indeed, our method
of projectability turns out to be our way of implementing the
covariance of the theory; in fact, it allows us to get rid of the
family of all observers, which is encoded in the quantum connection
(through $J_1\f E$).

The quantum connection allows us to take derivatives of sections $\psi\colon \f
E \to \f Q$. We have the expression:
\begin{equation}\label{eq:9}
\nab_\lam \psi  d^\lam \eqv
(\der_\lam \psi - \coi  \K q_\lam  \psi)  d^\lam =
(\der_0 \psi + \coi  \C H_0  \psi)  d^0 +
(\der_j \psi - \coi  \C P_j  \psi)  d^j ,
\end{equation}
and its `observed' counterpart
\begin{equation}\label{eq:10}
\ob\nab_\lam \psi \eqv (\der_\lam \psi - \coi  A_\lam  \psi),
\end{equation}
where the superscript $o$ means that the covariant derivative is related to the
pull-back connection $o^*\K q$. Using the splittings of the previous section we
may also define
\begin{equation}
\ba\nab \Psi =
 (\der_0 \psi + \dt x^j_0  \der_j \psi - \coi  \C L_0  \psi)
  d^0 \ten b ,
\qquad
\wch\nab \Psi =
 (\der_j \psi - \coi  \C P_j  \psi)  \ch d^j \ten b .\label{eq:11}
\end{equation}
Furthermore, given an observer $o$ we define the {\em observed quantum
  Laplacian\/} of  $\Psi \in \C S(\f Q)$ to be the section
$\ob\Del \Psi = \bar{G}(\check\nabla^o\check\nabla^o\Psi)$ with coordinate
expression in adapted coordinates
\beq
\ob\Del \Psi =
G^{hk}  \big((\der_h - \coi  A_h) (\der_k - \coi  A_k) +
K_h{}^l{}_k  (\der_l - \coi  A_l) \big) \psi  u^0 \ten b .
\eeq

J. Jany\v{s}ka \cite{Jan95b} has proved that all covariant
quantum Lagrangians of the quantum bundle are proportional to
\begin{equation}
\E L[\Psi] = \frac{1}{2} dt \wed
\big(h(\Psi, \imi\ba\nab \Psi) + h(\imi\ba\nab \Psi, \Psi)
- (\ba G \ten h)(\wch\nab \Psi, \wch\nab \Psi)
+ k r  h(\Psi, \Psi)\big)\eta ,\label{eq:12}
\end{equation}
with coordinate expression
\begin{multline*}
{\E L}[\Psi]
= \frac{1}{2}\big( \imi(\ba \psi\der_0\psi - \psi\der_0\ba \psi)
- G^{hk}_0\der_h\ba\psi\der_k\psi
\\
+ \imi G^{hk}_0 A_h(\psi\der_k\ba\psi - \ba\psi\der_k\psi)
+ \bar{\psi}\psi(2A_0 - G^{rs}_0A_rA_s + kr_0)
\big)\rtd g d^0 \wed d^1 \wed d^2 \wed d^3
\end{multline*}
where $k$ is an arbitrary real factor and $r =r_0u^0 : \f E \to \Rn \ten \B
T^*$ is the scalar curvature of the spacelike metric $G$. The corresponding
Euler--Lagrange expression is $h\Sha(\E E[\Psi]) : \f E \to \B L^3 \ten (\f Q
\uten{\f E} \Lam^4 T^*\f E)$, with coordinate expression
\begin{multline}\label{eq:13}
h\Sha(\E E[\Psi]) =
2 \big(\coi (\der_0 - \coi A_0 + \tfr12 \di 0 g)\psi +
\tfr12 G^{hk}_0(\der_h - \coi A_h)(\der_k - \coi A_k)\psi
\\
+ \tfr12 \diG h k g
(\der_k - \coi A_k)\psi + \tfr12 kr_0\psi \big)
b \ten \rtd g d^0 \wed d^1 \wed d^2 \wed d^3.\END
\end{multline}
J. Jany\v{s}ka and M. Modugno have proved a uniqueness-by-covariance result
for  $\E E$ \cite{JanMod02}. Thus, $k$ remains undetermined in our
scheme in contrast with other authors. See
Section~\ref{sec:energy-operator-from} for a more detailed discussion.

Next, we introduce a way to associate to each quantizable function $f$ a vector
field on the quantum bundle $\f Q$. More precisely, it is proved
\cite{JadMod94,JadJanMod98} that there is a natural Lie algebra isomorphism
between quantizable functions and a space of vector fields on $\f Q$ obtained
as follows. Given $f$ there is a unique Hermitian vector field $Y\Upa_f$ on the
extended quantum bundle $\f Q\Upa$ such that it is projectable to $J_1 \f E$,
it is $\K q$-horizontal and its covariant differential $\nabla[\K q]Y\Upa_f$
takes its values in the subbundle $\B T^* \ten \f Q$, in particular $Y\Upa = \K
q(H[f]) + \imi f$. The vector field $Y\Upa_f$ turns out to be projectable onto a
vector field $Y_f$ on the quantum bundle $\f Q$, which is said to be a
\emph{quantum vector field}. We have the coordinate expression
\begin{equation}\label{eq:15}
Y_f = f^0  \der_0 - f^j  \der_j + \imi (f^0  A_0 -
f^h  A_h + f_0)  z  \der z ,
\end{equation}
The space of quantum vector fields constitute a Lie algebra; it can be proved
that it is naturally isomorphic to the Lie algebra of quantisable functions.

The quantum vector field $Y_f$ acts on the sections $\Psi$ of the quantum
bundle via the associated Lie derivative $Z_f \byd \imi L_{Y_f} $. This is
possible since $Y_f$ is projectable. Note that the Lie derivative of $\Psi$
with respect to $Y=Y^\lambda\der_\lambda + Yz\der z$ is $L_Y \Psi =
Y^\lambda\der_\lambda\psi - Y\psi$. However, it is not enough to have operators
on sections of the quantum bundle, since in view of the probabilistic
interpretation of wave functions in quantum mechanics we should be able to
compute spacelike integrals of quantum sections. With this aim in mind we are naturally led to
the introduction of half-forms. These are geometric objects that can be paired each
other in order to yield densities. Such densities can be integrated in order to
define a Hilbert space norm on the space of quantum states. Namely, we
introduce the bundles over $\f E$
\begin{equation}
  \label{eq:16}
  \f Q^\eta \byd \f Q \ten \B L^{3/2}\ten
   \sqrt{\wedge^3 V^*\f E},\qquad
  \f Q^{\bar{\eta}} \byd \f Q \ten \B T^{1/2} \ten \B L^{3/2}\ten
   \sqrt{\wedge^4 T^*\f E}
\end{equation}
whose sections are said to be \emph{half-forms}.
Here the square root of an oriented vector space is the vector space whose
tensor square is the initial vector space, and the bases of the above square
roots are the square roots of the bases of the corresponding spaces,
\emph{i.e.} the square roots of the volumes. Note that the square root of the
volume element is parallel with respect to the spacetime connection $K\Nat$.  We will
use the notation $\Psi^\eta = \Psi\ten \sqrt{\eta} = \psi\sqrt[4]{|g|}b\ten
\sqrt{\check{d}^1\wedge \check{d}^2 \wedge \check{d}^3}$, and analogously for
$\Psi^{\bar{\eta}}$. We will also make use of the symbols $\psi^\eta =
\psi\sqrt[4]{|g|}$ and $v = \check{d}^1\wedge\check{d}^2\wedge\check{d}^3$, so
that $\Psi^\eta = \psi\sqrt[4]{|g|} b\otimes \sqrt{v}$. Note that the vertical
Riemannian connection induced by the metric $g$ yields a connection on the
bundle $\sqrt{\wedge^3 V^*\f E}\to \f E$: indeed, if $f\sqrt{\eta}$ is a
section of this bundle, then
\begin{equation}
  \label{eq:17}
  \check{\nabla} (f\sqrt{\eta}) = (\der_i f
  +\frac{1}{2}\Gamma^j_{ij})\check{d}^i\otimes\sqrt{\eta}.
\end{equation}
The observed Laplacian can be defined on half-forms using the tensor product
of the connection $o^*\K q$ with the above Riemannian connection on $\sqrt{\wedge^3
  V^*\f E}$.

Now, let us define the operator
\begin{equation}\label{eq:18}
Z_f(\Psi^\eta) \byd \imi L_{Y_f}(\Psi^{\bar{\eta}})
\otimes \frac{1}{\sqrt{\bar{\eta}}}\otimes\sqrt{\eta}.
\end{equation}
Note that we cannot compute directly the Lie derivative of $\sqrt{\eta}$ with
respect to a non-vertical vector field; so, we are forced to use
$\sqrt{\bar{\eta}}$ in an obvious way. Note that we have
\begin{equation}
  \label{eq:19}
  L_{Y_f}\sqrt{\bar{\eta}} =
  \frac{1}{2}\frac{\der_\lam(Y^\lam_f\sqrt{|g|})}{\sqrt{|g|}}
   \sqrt[4]{|g|} \sqrt{u_0\ten d^0\wedge d^1\wedge d^2\wedge d^3}
\end{equation}
We have the expression
\begin{multline}
  \label{eq:20}
  Z_f(\Psi^\eta) = i\left(f^0 \ob\nabla_0 -f^i\ob\nabla_i - if_0
  + \tfr12\Big(\frac{\der_0(f^0\sqrt{|g|})}{\sqrt{|g|}}
     - \frac{\der_i(f^i\sqrt{|g|})}{\sqrt{|g|}}\Big)\right)(\psi)
   \\
   \sqrt[4]{|g|}
     b\otimes\sqrt{v}.
\end{multline}
In particular, we obtain
\begin{gather}
Z_{x^\alp}(\Psi^\eta) = x^\alp  \Psi^\eta,\quad
Z[\C P_j](\Psi) = - i \der_j (\psi^\eta)b\otimes\sqrt{v},\quad
Z[\C H_0](\Psi) = i \der_0(\psi^\eta) b\otimes\sqrt{v}.
\end{gather}

As far as the Euler--Lagrange expression \eqref{eq:13} is concerned, we can
rewrite it using an observer $o$.  Namely, we have the vector field $X^o =
o\lrcorner o^*\K q$ and the equality
\begin{equation}
  \label{eq:14}
  h\Sha(\E E[\Psi]) = 2\left(i \frac{L_{X^o}\Psi^{\bar{\eta}}}{\sqrt{\bar{\eta}}}
   +\frac{1}{2}\overset{o}{\Delta}\Psi + \frac{1}{2}kr\Psi\right)
    \otimes\bar{\eta}.
\end{equation}
The above expression yields an operator on half-forms in a natural way. Indeed
the first summand is just the Lie derivative of a half-form and we have the
equality
\begin{equation}\label{eq:21}
(\overset{o}{\Delta}\Psi)\otimes\sqrt{\eta} = \overset{o}{\Delta}(\Psi^\eta)
\end{equation}
since $\sqrt{\eta}$ is parallel with respect to the connection
induced on the half-forms bundle. For this reason we define the
\emph{Schr\"odinger operator} to be the operator $\E S$
\begin{equation}
  \label{eq:22}
  \E S(\Psi^\eta) = - \frac{i}{2} h\Sha(\E E[\Psi]) \otimes \frac{1}{\bar{\eta}}
  \otimes \sqrt{\eta} = L_{X^o}\Psi^\eta -
  \frac{i}{2}\overset{o}{\Delta}(\Psi^\eta) - \frac{1}{2}kr\Psi^\eta.
\end{equation}

Next, we consider the pre--Hilbert {\em functional quantum bundle\/}
$\f H \to \f T$ over time. This is defined as follows: for each $\tau\in \f T$
let
\begin{equation}
H_\tau\byd \{\Psi^\eta_\tau\colon \f E_\tau \to \f Q^\eta_\tau \mid \Psi^\eta_\tau =
\Psi^\eta|_{E_\tau},\ \Psi^\eta\ \text{quantum section with compact support}\}
\label{eq:23}
\end{equation}
where $\f E_\tau = t^{-1}(\tau)$.
In other words, the infinite dimensional fibres are constituted by the sections
of the quantum bundle at a given time and with compact support. The space $\f
H$ can be given the structure of an $F$-smooth manifold, in the sense
of~\cite{Fro82}. The functional quantum bundle also inherits the Hermitian
structure:
\bEq
\hat{h} : \f H \ucar{\f T} \f H \to \B C :
(\Psi^\eta_\tau,\Psi^\eta{}'_\tau) \mto \int_{\f E_\tau}
h(\Psi^\eta_\tau, \Psi^\eta{}'_\tau) ,
\eEq
which makes it a pre-Hilbert bundle.

The tangent space is defined to be the set $T\f H\byd
\cup_{\Psi^\eta_\tau \in \f H} T_{\Psi^\eta_\tau}\f H$ where
\begin{equation}
  \label{eq:24}
  T_{\Psi_\tau}\f H =\{ \zeta_{(\tau,u)}\colon T_{(\tau,u)}\f E \to T_{(\tau,u)} \f Q
  \mid T\zeta_{(\tau,u),e} = V_e\Psi^\eta_\tau\}.
\end{equation}
In other words, the coordinate expression of $\zeta_{(\tau,u)}$ is
$\zeta_{(\tau,u)} = \zeta + \der_i \psi \dot x^i$, where $\zeta$ is a
complex-valued function on $\f E$. Let us recall that any section
$\Psi^\eta\colon \f E \to \f Q$ (which is defined on a tube-like open subset)
yields the $F$-smooth section
\begin{equation}
\hat\Psi^\eta\colon \f T\to \f H,\quad
\hat\Psi^\eta(\tau)(e_\tau)=\Psi^\eta(e_\tau);
\label{eq:25}
\end{equation}
conversely, every $F$-smooth section of $\f H \to \f T$ yields a section
$\Psi^\eta$ as above, establishing a bijective correspondence.

A connection on the space $\f H$ can be introduced as a section $\chi\colon \f
H \ucar{\f T} T\f T \to T\f H$ which is linear over $\f H$ and projects onto
$\id_{T\f T}$. A connection $\chi$ acts on sections as
$\chi(\hat\Psi^\eta)=\chi \circ \hat\Psi^\eta$, with coordinate expression
$\chi(\hat\Psi^\eta) = \chi_0(\psi^\eta)u^0 + \der_i\psi^\eta \dot x^i$. The
covariant differential of sections is defined by
\begin{equation}
  \label{eq:26}
  \nabla[\chi]\hat\Psi^\eta = T\hat\Psi^\eta - \chi(\hat\Psi^\eta)
\end{equation}
with coordinate expression
$\nabla[\chi]\hat\Psi^\eta = \der_0 \psi^\eta - \chi_0(\psi^\eta)u^0$.

It is now obvious that the Schr\"odinger operator $\E S$ is the
covariant differential $\nab[\chi]$ of a connection $\chi$ on the functional
quantum bundle; hence, the quantum Lagrangian yields a lift of the
quantum connection $\K q$ of the extended quantum bundle to a
connection $\chi$ of the functional quantum bundle. The coordinate expression
of $\chi$ is
\begin{equation}\label{eq:27}
\ke_0(\Psi^\eta)
= \left(\coi A_0 + \frac{i}{2}\ob\Del_0 + \frac{1}{2}kr_0 \right)\Psi^\eta.
\end{equation}

Let us consider a quantisable function $f$. The operator $Z_f$ can be defined
on sections of the functional quantum bundle in an obvious way. Then
\begin{equation}
\hat{f} = Z_f - i f^0 \con \nab[\chi]\label{eq:28}
\end{equation}
is the unique combination of $Z_f$ and $\nab[\chi]$ which yields an
operator acting on the fibres of the functional quantum bundle. We
have the following coordinate expression
\begin{equation}\label{eq:29}
\hat{f}(\Psi^\eta)
= \left(- \frac12  f^0   \ob\Del{}_0 - \imi  f^j  \ob\nab_j
+ f_0 - \frac12  k  f^0  r_0
- \imi  \frac12  \frac{\der_j (f^j  \rtd g)}{\rtd g} \right) 
\Psi^\eta.
\end{equation}
The map
$f \mto \hat{f}$ is injective. Moreover, $\hat{f}$ is Hermitian.
We assume $\hat{f}$ to be the Hermitian {\em quantum operator\/}
associated with the quantisable function $f$. This is our {\em
correspondence principle\/}.

\bEx
Let us consider an observer $o$ and a time scale $u_0$ and let us
refer to a chart adapted to the observer and to the time scale. Then,
the quantum operators associated with the quantisable functions
$x^\alp, x^i_0, \C P_i, \C H_0$ are given, for each
$\Psi \in \C S(\f Q)$,
by
\begin{gather}\label{eq:30}
\widehat{x^\alp}(\Psi^\eta) = x^\alp  \Psi^\eta ,\qquad
\widehat{\C P_j}(\Psi^\eta) =
- i \der_j(\psi^\eta) b\otimes\sqrt{v}  ,
\\
\begin{split}
\widehat{\C H_0}(\Psi^\eta) =
\left(- \frac12 G^{hk}  \big((\der_h - \coi  A_h)
(\der_k - \coi A_k) + K_h{}^l{}_k  (\der_l - \coi  A_l)\big)(\psi)\right.
\\
\left.\hphantom{\widehat{\C H_0}(\Psi^\eta) =}
 - A_0\psi - \frac12 k r_0\psi\right) \sqrt[4]{|g|} b\otimes\sqrt{v}.\END
\end{split}
\end{gather}
\eEx

The commutator of Hermitian fibred operators on the functional quantum bundle
yields a Lie algebra structure. However from the formula
\begin{equation}
[\hat{f} , \hat{g}] = \widehat{\db[f , g\db]}
+ \big[(g'' \ten L_{Y_f} - f'' \ten L_{Y_g}) , \E S\big]
.\label{eq:31}
\end{equation}
we obtain that the correspondence principle fails to be a Lie algebra morphism
exactly on quantisable functions with nontrivial quadratic term.

\bRm
The Feynmann path integral formulation of Quantum Mechanics
 can be naturally expressed in our formalism; in
particular, the Feynmann amplitudes arise naturally via parallel
transport with respect to the quantum connection \cite{JadMod94}. So
the Feynmann path integral can be regarded as a further way to lift
the quantum connection $\K q$ to a functional quantum connection.
\eRm

\bRm
In the particular case when spacetime is flat, our quantum
dynamical equations turns out to be the standard Schr\"odinger
equation and our quantum operators associated with spacetime
coordinates, momenta and energy coincide with the standard operators.
Therefore, all usual examples of standard Quantum Mechanics are
automatically recovered in our covariant scheme.
\eRm

\bRm
The above procedure can be easily extended to classical and quantum
multi--body systems (\emph{e.g.}, the rigid body, see
\cite{ModTejVit08a,ModTejVit08b}), to particles with spin (Pauli equation
\cite{CanJadMod95}), and to a more limited extent to the Einstein relativistic
mechanics \cite{JanMod08}.
\eRm

\section{Energy operator from CQM to GQ}
\label{sec:energy-operator-from}
In this section we now restrict ourselves to the case when $\f E = \f T\car \f
M$, where $\f M$ is an orientable Riemannian manifold. Here, in principle, the
theory allows a time-dependent metric, but we will not consider this general
situation. Note that $J_1(\f T \car \f M)=\f T\car T\f M$. Our task is to
assume the same structures of the GQ theory in the case of a particle (or a
`generalized' particle in the case in which $\dim \f M \neq 3$) and compare the
energy operator from CQM with the one obtained in GQ. We will refer to
\cite{Sni80} for a detailed derivation of the energy operator in this situation
(see p.\ 120, Section 7.2, or p.\ 180, Section 10.1 for the case with a nonzero
electromagnetic field).

We require the gravitational field to be purely space-like, \emph{i.e.}
$K\Nat{}_0{}^h{}_j = K\Nat{}_j{}^h{}_0 = 0$, $K\Nat{}_0{}^h{}_0 = 0$; this
request is intrinsic in view of the splitting $\f E = \f T\car \f M$ of
spacetime. We assume $F=0$, even if we could at least consider a nonzero
magnetic field in principle.

Then, there is a natural symplectic structure on $T\f M$ which is the
pull-back of the canonical structure on $T^*\f M$ under the metric
isomorphism $g\Fla$.  The form $\Omega$ reduces to the above symplectic form,
and the second-order connection $\gamma$ is the usual geodesic spray
$\gamma\colon T\f M \to T^2\f M$. The classical theory can be completely
developed from the previous assumptions.

Since we are in a time independent situation, the CQM theory can be developed
by assuming a quantum bundle $\f Q \to \f T\car\f M$ which is the pull-back of
a Hermitian complex line bundle $\bar{\f Q}\to\f M$. In fact, it can be proved
that all quantum bundles are of this form. In the same way, if we consider the
line bundle $\bar{\f Q}\Upa=\tau_{\f M}^*\bar{\f Q}\to T\f M$, where $\tau_{\f
  M}\colon T\f M \to \f M$ denotes the tangent bundle projection, then the
extended quantum bundle ${\f Q}\Upa\to \f T\car T\f M$ is obtained by pulling
back $\bar{\f Q}\Upa\to T\f M$ to $\f T\car T\f M$. Moreover, any quantum connection
$\K q$ on ${\f Q}\Upa\to \f T\car T\f M$ that fulfills the curvature identity
$R[{\K{q}}] = \imi \Ome\ten\id_{\f Q}$ is obtained in a two step process. In
the first one we consider the connection $\bar{\K q}\Upa$ on $\bar{\f Q}\Upa\to
T\f M$ obtained by pulling back a Hermitian connection $\bar{\K{q}}$ on $\bar{
  \f Q} \to \f M$ and adding to it a suitable constant multiple of the $1$-form
$(g\Fla)^*\theta$, where $\theta$ is the Liouville form on $T^*\f M$, see \cite{lopez3} for the precise details. In the
second step, we pull back $\bar{\K{q}}$ to ${\f Q}\Upa\to \f T\times T\f M$ in
order to get the connection ${\K{q}}$. It is clear that $\bar{\f Q}\Upa\to T\f
M$ and $\bar{\K q}\Upa$ define a quantum structure in the sense of the standard
GQ.

The polarization $\f P$ that we choose is the vertical one, \emph{i.e.}  $\f
P=VT\f M=\ker\tau_{T\f M}\subset TT\f M$ which is locally spanned by the vector
fields $\partial/\partial \ddot x^i$.  We have the canonical isomorphism $VT\f
M\simeq T\f M \underset{\f M}{\times} T\f M=\tau_{\f M}^*T\f M$.  We stress
that half-forms in CQM and half-forms in GQ in the above setting are the same
(up to dependency on time, which in GQ is not explicit). Indeed, the
determinant bundle of the polarization $\f P$ is by definition $K_{\f
  P}=\det(\f P^\circ)$, where $\f P^0\subset T^*T\f M$ is the subbundle that
anihilates $\f P\subset TT\f M$. One has the following exact sequence
\begin{displaymath}
0\to \f P\to T T\f M\to \tau_{\f M}^* T\f M\to 0,
\end{displaymath}
and taking duals we get the exact sequence
\begin{displaymath}
0\to \tau_{\f M}^*T^*\f M\to T^*T\f M\to \f P^*\to 0.
\end{displaymath}
This shows that $\f P^\circ= \tau_{\f M}^*T^*\f M$ and therefore
\begin{displaymath}
K_{\f
  P}=\det(\tau_{\f M}^*T^*\f M)=\tau_{\f M}^*\det(T^*\f M)=\tau_{\f
  M}^*\Lambda^n T^*\f M,
\end{displaymath}
where $n$ is the dimension of $\f M$. Since $\f M$ is orientable, it admits a
natural metalinear structure that allows us to construct the bundle $N_{\f
  P}^{1/2}$ of half-forms normal to $\f P$ which is a square root of the
canonical bundle; i.e. $N_{\f P}^{1/2}=\sqrt{\tau_{\f M}^*\Lambda^n T^*\f
  M}=\tau_{\f M}^*\sqrt{\Lambda^n T^*\f M}$. The Lie derivative of forms
induces on the canonical bundle $K_{\f P}$ a partial covariant derivative
$\nabla$ defined along $\f P$. This, in turn, induces on $N_{\f P}^{1/2}$ a
natural partial covariant derivative $\nabla^{1/2}$ defined along $\f
P$. Therefore we can endow the line bundle $\bar{\f Q}\Upa\otimes N_{\f
  P}^{1/2} $ with the tensor product partial covariant derivative
$\nabla'=\nabla[\bar{\K q}\Upa]\otimes 1+1\otimes \nabla^{1/2}$ defined along
$\f P$. The space of $\f P$-polarized sections of the line bundle $\bar{\f
  Q}\Upa\otimes N_{\f P}^{1/2}$ is
\begin{displaymath}
\Gamma_{\f P}(\bar{\f Q}\Upa\otimes N_{\f
  P}^{1/2})=\{\xi\in \Gamma(T\f M,\bar{\f Q}\Upa\otimes N_{\f P}^{1/2})\colon
\nabla'_V\xi=0,\ \forall\ V\in\Gamma(T\f M,\f P)\}.
\end{displaymath}
The Hilbert space of GQ quantum states is given by the $L^2$-completion of
$\Gamma_{\f P}(\bar{\f Q}\Upa\otimes N_{\f P}^{1/2})$. Since the line bundle
$\bar{\f Q}\Upa\otimes N_{\f P}^{1/2} =\tau_{\f M}^*(\bar{\f Q}\otimes
\sqrt{\Lambda^n T^*\f M})\to T\f M$ is a pullback, its space of global sections
is given by
\begin{displaymath}
\Gamma(T\f
M,\bar{\f Q}\Upa\otimes N_{\f P}^{1/2})=C^\infty(T\f M)\otimes_{C^\infty(\f M)}
\Gamma(\f M,\bar{\f Q}\otimes \sqrt{\Lambda^n T^*\f M}).
\end{displaymath}
Taking into account now that the partial covariant derivative $\nabla'$ is also
a pullback, we immediately obtain the identification
\begin{displaymath}
\Gamma_{\f P}(T\f M,\bar{\f Q}\Upa\otimes N_{\f P}^{1/2})= \Gamma(\f M,\bar{\f
  Q}\otimes \sqrt{\Lambda^n T^*\f M}).
\end{displaymath}

The quantum operators on half-forms corresponding with position and momentum
observables are just the same as~\eqref{eq:30}; compare it with eq.\ 7.82, p.\
128 of \cite{Sni80}. The difference between CQM and GQ lies in the way how the
energy is quantized.

Concerning the energy $\C H_0$, from \eqref{eq:29} we have the expression
\begin{equation}
  \label{eq:32}
  \hat{\C H}_0(\Psi^\eta)
= \left(- \frac12 \ob\Del{}_0 - A_0 - \frac12  k  f^0  r_0 \right)
\Psi^\eta,
\end{equation}
and we can use \eqref{eq:21} in order to write the first summand as $- \frac12
\ob\Del{\Psi}_0\otimes\sqrt{\eta}$. In this way we realize that the only
difference between the above formula and the corresponding formulae 7.114 on
p.\ 134 ($F=0$) and 10.59 on p. 180 ($F\neq 0$) of \cite{Sni80} is the factor
in front of the scalar curvature $r$\footnote{The difference in sign is due to
  the fact that, like in \cite{Wu}, we use the opposite convention about the
  value of $r$ with respect to \cite{Sni80}}.

In order to decide which factor must be used for introducing the scalar
curvature in the energy operator we recall that in GQ the quantum operator
corresponding to the kinetic energy is obtained by the
Blattner-Kostant-Sternberg (BKS) method. This method is useful for those
observables $f$, like energy, whose Hamiltonian vector field $X_f$ does not
preserve the polarization. The value of the corresponding quantum operator on a
wave function $\Psi$ is obtained by dragging $\Psi$ using the flow of the lift
of $X_f$ to the half-form bundle and then projecting the result back to the
space of polarized sections. Following \cite[eq.\ 10.52]{Sni80}, the
computation of the energy operator by the BKS method yields at the pole of a
normal coordinate system (where $\der_ig_{jk}=0$)
\begin{equation}
  \label{eq:33}
  \C Q (\psi^\eta b\otimes\sqrt{v}) = \ob\Delta(\psi^\eta b)\otimes\sqrt{v}
\end{equation}
We stress that here $\sqrt{v}$ is just a local basis of the bundle of volume forms, and not
the global section $\sqrt{\eta}$ (see~\eqref{eq:16} and the sentences
thereafter). Then a computation at the pole of a normal coordinate system shows
that
\begin{equation}
  \label{eq:34}
  \ob\Delta(\psi^\eta b)\otimes\sqrt{v}  =(\ob\Delta(\Psi)
  + \frac{1}{6} r \Psi) \sqrt{\eta}
\end{equation}
(see also \cite{Woo92}). The right-hand side of the above formula is a globally
defined tensor, while it is not possible to interpret the left-hand side as an
intrinsic expression by means of the available connections.
\begin{Lemma}
  The following equality holds:
  \begin{displaymath}
    \ob\Delta(\psi^\eta b)\otimes\sqrt{v} = \ob\Delta(\Psi^\eta)
    - 2\bar{G}(\check\nabla^o(\psi^\eta b)\otimes\check\nabla^o\sqrt{v})
    - \psi^\eta b \otimes\bar{G}(\check\nabla^o\check\nabla^o\sqrt{v})
  \end{displaymath}
\end{Lemma}
\begin{proof}
  Indeed we have
  \begin{align}
  \label{eq:359}
  \ob\Delta(\Psi^\eta) &= \bar{G}(\check\nabla^o\check\nabla^o
  (\psi^\eta b\otimes\sqrt{v}))
  \\
  &=\bar{G}(\check\nabla^o\check\nabla^o(\psi^\eta b))\otimes\sqrt{v}
  + 2\bar{G}(\check\nabla^o(\psi^\eta b)\otimes\check\nabla^o\sqrt{v})
  + \psi^\eta b \otimes\bar{G}(\check\nabla^o\check\nabla^o\sqrt{v}),
  \end{align}
  and the statement is proved by observing that
  \begin{displaymath}
    \ob\Delta(\psi^\eta b)\otimes\sqrt{v} =
    \bar{G}(\check\nabla^o\check\nabla^o(\psi^\eta b))\otimes\sqrt{v}.
  \end{displaymath}
\end{proof}
\begin{Theorem}
  At the pole of a normal coordinate system we have
  \begin{displaymath}
    \ob\Delta(\Psi^\eta) = \ob\Delta(\psi^\eta b)\otimes\sqrt{v}
    - \frac{1}{6} r \Psi^\eta.
  \end{displaymath}
\end{Theorem}
\begin{proof}
   Using the above Lemma we have
\begin{align}
  \begin{split}\label{eq:36}
  \ob\Delta(\Psi^\eta)&=\ob\Delta(\psi^\eta b)\otimes\sqrt{v}
  + {G}^{ij}(\check\nabla^o_i(\psi b)\sqrt[4]{|g|}
    \otimes\Gamma^h_{jh}\sqrt{v})
  +{G}^{ij}\psi b\der_i\sqrt[4]{|g|}
    \otimes\Gamma^h_{jh}\sqrt{v}
  \\
  &\hphantom{={}} +\psi^\eta b
    \otimes{G}(\check\nabla^o(\frac{1}{2}\Gamma^{h}_{ih}
    \check{d}^i\otimes\sqrt{v}))
  \end{split}
  \\
  \begin{split}\label{eq:37}
  &=\ob\Delta(\psi^\eta b)\otimes\sqrt{v}
  + {G}^{ij}(\check\nabla^o_i(\psi b)\sqrt[4]{|g|}
    \otimes\Gamma^h_{jh}\sqrt{v})
  \\
  &\hphantom{={}}
  +\frac{1}{2}\psi^\eta b\otimes{G}^{ij}(\der_i\Gamma^{h}_{jh}
    + \Gamma^{k}_{ij}\Gamma^{h}_{kh} - \frac{1}{2}
    \Gamma^{k}_{ik}\Gamma^{h}_{jh})
    \sqrt{v}))
  \end{split}
\end{align}
At the pole of the normal coordinate system we have $\Gamma^i_{jk} = 0$, hence
\begin{align}
  \label{eq:38}
  \ob\Delta(\Psi^\eta)&=\ob\Delta(\psi^\eta b)\otimes\sqrt{v} +
  \frac{1}{2}\psi^\eta b\otimes{G}^{ij}\der_i\Gamma^{h}_{jh}\sqrt{v};
\end{align}
it can be proved that $r=\tfrac{3}{2}G^{ik}g^{pq}\der_i\der_k g_{pq}$ and that
$\tfrac{1}{2}{G}^{ij}\der_i\Gamma^{h}_{jh} =
- \tfrac{1}{4}G^{ij}g^{hk}\der_i\der_j g_{hk}$, so that the statement is proved.
\end{proof}

So, if, according to CQM, we define the energy operator using the Bochner
Laplacian which takes into account the Riemannian connection on the square root
bundle $\sqrt{\wedge^3 T^*\f M}\to \f M$ then the scalar curvature term which
arises in GQ is canceled by a similar term arising from the covariant derivative
of the base of sections of the square root bundle.

On the other hand, while the `intermediate' term $\ob\Delta(\psi^\eta
b)\otimes\sqrt{v}$ obtained by the BKS method has no intrinsic meaning, the
final result (\emph{i.e.} the right-hand side of \eqref{eq:34}) is intrinsic.
CQM allows us to recover this term by adding a term with the scalar curvature
multiplied by an arbitrary coefficient, using our quantum Lagrangian approach.
So we can, in a sense, recover the expression of \cite{Sni80} through CQM.

\section{Conclusions}

We have discussed the two ways for defining the quantum energy operator proposed by CQM and GQ. The energy operators obtained by these two methods differ by a multiplication operator by a constant times the scalar curvature. This constant can be arbitrarily modified if one uses the
Lagrangian approach, or even completely removed if one uses covariant
derivatives of half-forms.

It is a well-known feature of GQ that non-trivial examples are very few since
it is very easy to run into topological obstructions and several other complications.
In all known examples of GQ whose spectral problem for the energy operator has been analyzed, the scalar curvature term is just zero or a constant. Among the latter ones  we can mention the results obtained for the Landau problem on Riemann surfaces \cite{lopez3, Tej99, LopTej01, Tej06-1, Tej06} and  for the rigid body  \cite{Tej00, Tej01}. In these cases the spectrum of the Schr\"odinger operator is modified by an overall shift.

At the moment, we can only say that the possibility that the scalar curvature
plays no r\^ole in quantum mechanics is not remote. One possibility is that one might be
able to modify the BKS method in such a way as to incorporate the action of the
Riemannian connection on the square root bundle. In general, most well known
polarizations are endowed with a fibre metric and therefore in principle it should be possible
to define a Riemannian connection acting on 
half-forms. Another possibility is that the whole BKS procedure could be re-expressed in an
`infinitesimal' way through the parallel transport of the connection $\chi$ on
the infinite-dimensional bundle. It is known \cite{JadMod92,JadMod93,JadMod94}
that such a parallel transport leads to the Feynman integral formulation, and this could also be
useful in order to perform a covariant analysis of De Witt's approach
\cite{dew}.

We hope to solve the problem of scalar curvature in quantum mechanics in a
future research.  

\fz

\end{document}